\documentclass[runningheads, 11pt]{llncs}
\usepackage{amsfonts}
\usepackage{cite}
\usepackage{amsmath}
\usepackage{cases}
\usepackage{color}
\usepackage{amssymb}
\usepackage{graphicx}
\usepackage{ctex}
\usepackage{lineno}
\def\fqu#1 {\fbox {\footnote {\ }}\ \footnotetext { From Qu: {\color{red}#1}}}
\def\fxh#1 {\fbox {\footnote {\ }}\ \footnotetext { From Hai: {\color{blue}#1}}}

\usepackage{url}
\urldef{\mailsa}\path|xiong.hai@163.com, ljqu_happy@hotmail.com, lichao_nudt@sina.com|

\begin{document}
\begin{CJK*}{GBK}{song}
\mainmatter  

\title{A New Method to Compute the 2-adic
Complexity of Binary Sequences}

\titlerunning{2-adic complexity}

%

\author{Hai Xiong\and Longjiang Qu\and Chao Li}
\authorrunning{Xiong Hai, Longjiang Qu, Chao Li}

\institute{College of Science,
National University of Defense Technology, Changsha  410073, China
\mailsa\\
}

\toctitle{Lecture Notes in Computer Science}
\tocauthor{Authors' Instructions}
\maketitle

\begin{abstract}
In this paper, a new method is presented to compute the $2$-adic complexity of pseudo-random sequences.
With this method, the $2$-adic complexities of
 all the known sequences with ideal $2$-level autocorrelation are uniformly determined.
 Results show that their $2$-adic complexities equal their periods. In other words, their $2$-adic complexities attain the maximum.
Moreover, $2$-adic complexities of two classes of optimal autocorrelation sequences with period $N\equiv1\mod4$, namely Legendre sequences and Ding-Helleseth-Lam sequences,
 are investigated.
Besides, this method also can be used to
compute the linear complexity of binary sequences regarded as sequences over other finite fields.

\end{abstract}

\section{Introduction}
Linear feedback shift registers (LFSRs) and feedback
with carry shift registers (FCSRs) are two classes of pseudo-random sequence generators.
The sequences produced by them could have good randomness, such as low correlation, long period and so on.
These pseudo-random sequences  are widely used in cryptography and communication systems.

For any binary periodic sequence $s$, it always can be generated
by an LFSR or an FCSR. The length of the shortest LFSR resp.
FCSR which can generate $s$
 is called the linear complexity resp. $2$-adic complexity of $ s $, symbolically $LC( s )$
 resp. $AC( s )$. Since $ s $ can be completely determined by the Berlekamp-Massey
 algorithm \cite{Massey:1969} resp. rational approximation algorithm \cite{Klapper:1995} with $2LC( s )$ resp. $2AC( s )$
 consecutive bits, linear complexity and $2$-adic complexity are two of the most
 important security criteria of binary sequences.

 It is of interest to investigate the relationship between linear complexity and
 $2$-adic complexity.
 However, it may be quite difficult in general since little is known in the literature.
 Hence a natural tradeoff  is to investigate the linear complexity of sequences whose $2$-adic complexity is known or the $2$-adic
 complexity of sequences whose linear complexity is known.
 Until now, there are only a few  classes of pseudo-random sequences whose linear
 complexity
 and $2$-adic complexity both are clear.
  Seo et. al. \cite{Seo:2000} and Qi et. al. \cite{Qi:2003} got a lower bound on the linear complexity of a special class of
$l$-sequences respectively.  Klapper and Goresky \cite{Klapper:1997} derived a simple result about the $2$-adic complexity of $m$-sequences.
A breakthrough of this problem was given by Tian et. al. \cite{Tian:2010}.
   They completely determined the $2$-adic
 complexity of $m$-sequences and showed that all the $m$-sequences have optimal $2$-adic complexity.

$m$-sequences is a class of ideal $2$-level autocorrelation sequences which play a significant
role in applications for their optimal autocorrelation.
A large amount of ideal $2$-level autocorrelation sequences other than $m$-sequences have been constructed, for example
Legendre sequences, twin-prime sequences  and Hall's sextic residue sequences \cite{Golomb:2005}.
The linear complexities of these sequences have
all been determined; see \cite{Xiong:2013} for a survey. However, as
far as the authors known,
no result about the $2$-adic complexities of these sequences other than $m$-sequences is known yet.

In this paper, we will  present
a new method  to compute  the $2$-adic complexity of binary sequences. According to \cite{Klapper:1997}, to determine $2$-adic complexity of a binary sequence is equivalent to determine the greatest common divisor of two numbers which are associated with the sequence. Here, we convert this problem to compute   the determinant of a {circulant} matrix
 and the greatest common divisor of two other integers.
Then by using the new method,  we prove that all the known sequences with ideal $2$-level
  autocorrelation have the maximum $2$-adic complexities, i.e.
 their $2$-adic complexities  equal their periods. We also prove that Legendre sequences and Ding-Helleseth-Lam sequences with period $N\equiv1\mod4$
 have maximum $2$-adic complexities.
Hence Legendre sequences, twin-prime sequences and Hall's sextic residue sequences are nontrivial binary sequences whose
 linear complexities and $2$-adic complexities both could attain the maximum. Finally, as a byproduct,
 we show that the new method can be used to compute the  linear complexity
 of binary sequences when we regard them as sequences over other finite fields.

 The rest of this paper is organized as follows. Section 2 introduces some well-known results and notations.
 In Section 3,  a new method is presented to compute $2$-adic complexity of  binary sequences. In Section 4, $2$-adic complexities of ideal $2$-level autocorrelation sequences and
   two other classes of optimal autocorrelation sequences are determined.
 Section 5 presents some results on the linear complexity when one regard a  binary sequence
 as a sequence over another finite field.
 We conclude this paper in Section 6.

\section{Preliminaries}

In this section, we will introduce some notations and review some well-known results.
\subsection{Notations}
\begin{enumerate}
\item The symbol ``$+$" has a multiple meaning: it stands for the integer addition, or for the addition over $\mathbb{F}_2$, or even for  the addition over integer residue rings. But this will not bring confusion in concrete situations.
\item A sequence is called binary if its elements consist of $0$ and $1$.
\item For a binary sequence $s=(s_0, s_1,\cdots, s_{N-1})$, its \emph{sequence polynomial}  is $P_s(x)=\sum\limits_{i=0}^{N-1}s_ix^i$. The \emph{complementary sequence} of $s$, denoted by $\overline{s}$, is defined as $(1-s_0, 1-s_1, \cdots, 1-s_{N-1})$.
Let $D_s$ denote the \emph{support set} of $s$, which is defined as $D_{s}=\{0\le i\le N-1:\ s_i=1\}$.
\item Let $\mathbb{Q}_2$ denote the complete field of $\mathbb{Q}$ with respect to the $2$-adic absolute value.
\item Assume that $p=df+1$ is a prime. Let $\alpha$ be a primitive element of $\mathbb{F}_p$.
The \emph{cyclotomic classes} of order $d$ with respect to $\mathbb{F}_p$, denoted by $D_i^{(d,p)}(0\le i\le d-1)$, are defined as $D_i^{(d,p)}=\{\alpha^{i+kd}:0\le k\le f-1\}$.
\item Let   $S$  be a subset of $\mathbb{Z}/N\mathbb{Z}$ with $k$ elements. For any integer $\tau$, define $S+\tau=\{(a+\tau)\mod N:a\in S\}$. If there exists a positive integer $\lambda$ such that $|S\cap(S+\tau)|=\lambda$ for any $\tau\not\equiv0\mod N$, then $S$ is called an $(N, k, \lambda)$ \emph{cyclic difference set}.
 \end{enumerate}

\subsection{Optimal autocorrelation sequences}

Let $s =(s_0,s_1,\cdots,s_{N-1})$ be a binary sequence with period $N$.
The \emph{autocorrelation function} of $s$ is defined by
$$C_{s}(\tau)=\sum_{i=0}^{N-1}(-1)^{s_i+s_{i+\tau}},\ \tau\in \mathbb{Z}/{N\mathbb{Z}}.$$
Clearly, $C_{s}(0)=N$.

We say that $s$ is an \emph{optimal autocorrelation sequence} if
for any $\tau\ne0$, \\
(1)  $C_{s}(\tau)=-1$ and $N\equiv -1\mod4$; or\\
(2)  $C_{s}(\tau)\in\{1, -3\}$ and $N\equiv 1\mod4$; or\\
(3)  $C_{s}(\tau)\in\{2,-2\}$ and $N\equiv 2\mod4$; or\\
(4)  $C_{s}(\tau)\in\{0, -4\}$ and $N\equiv 0\mod4$.\\
In Case (1), the sequences are also said to have \emph{ideal $2$-level autocorrelation}.
{Many classes of such sequences have been reported, such as Legendre
sequence, Hall's sextic residue sequence, twin-prime sequence, and $m$-sequence, GMW
sequences, Maschiettie's hyperoval sequences, etc.
For a list of such sequences and detailed definitions of these sequences,
please refer to \cite{Nam:2006} or \cite{Cai:2009}. The following characterization of such sequences is from \cite{Golomb:2005}.
\begin{lemma}\label{le1} \cite{Golomb:2005}
Let $s$ be a binary ideal $2$-level autocorrelation sequence with period $N$.
Then $D_s$, the support set of $s$, is an $(N, \frac{N+1}{2}, \frac{N+1}{4})$ or $(N,\frac{N-1}{2},\frac{N-3}{4})$ cyclic difference set.
Based on their periods, all the known ideal $2$-level autocorrelation sequences
can be divided into three classes:
\emph{(1)} $N=2^n-1$;
\emph{(2)} $N=p$, where $p\equiv3\mod4$ is a prime;
\emph{(3)} $N=p(p+2)$, where both $p$ and $p+2$ are primes.
\end{lemma}

All the known binary sequences with optimal autocorrelation {until 2009} are surveyed by Cai and Ding \cite{Cai:2009}. Here
we only recall } the definitions of Legendre sequences and Ding-Helleseth-Lam sequences with
period of $N\equiv1\mod4$ {for later use}.

{\bf Legendre Sequences}: Let $p\equiv1\mod 4$ be a prime. Let $s$ be a binary sequence defined by
\begin{equation}
s_i=\left\{
\begin{aligned}
&1,\ \text{if } i\in D_0^{(2,p)};\\
&0,\ \text{otherwise } .
\end{aligned}
\right.\nonumber
\end{equation}
Then $s$ has optimal out-of-phase autocorrelation values ${\{1, -3\}}$.

{\bf Ding-Helleseth-Lam Sequences}: Let $p\equiv1\mod 4$ be a prime. Let $s$ be a binary sequence defined by
\begin{equation}
s_i=\left\{
\begin{aligned}
&1,\ \text{if } i\in D_0^{(4,p)}\cup D_1^{(4,p)};\\
&0,\ \text{otherwise } .
\end{aligned}
\right.\nonumber
\end{equation}
Then $s$ has optimal out-of-phase autocorrelation values ${\{1, -3\}}$.

\subsection{Feedback with carry shift register}

A \emph{feedback with carry shift register} (FCSR) consists of a feedback register and a memory cell.
It is designed by Klapper and Goresky \cite{Klapper:1993}.
The form of an $r$-stage FCSR is presented in Fig. \ref{f1}, where $q_i\ (1\le i\le r-1)\in \{0, 1\}$, $q_r=1$.
\begin{figure}[]
\begin{center}
\includegraphics[width=0.6\textwidth]{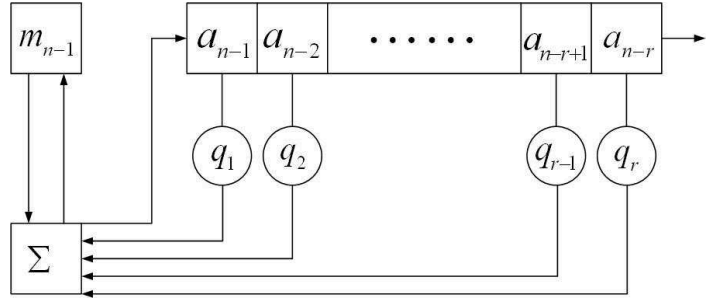}
\end{center}
\caption{Feedback with carry shift register}\label{f1}
\end{figure}
We call $q=\sum_{i=1}^{r}q_i2^i-1$  the \emph{connection number} of this FCSR
and its  operation is defined as follows:
\begin{enumerate}
\item Give an initial state $(a_{r-1}, a_{r-2},\cdots, a_0)$ of the register and $m$ of the memory,
where $a_i\in\{0, 1\}$, $m\in\mathbb{Z}$;
\item Compute an integer sum $\sigma=\sum_{i=0}^{r-1}q_ia_i+m$;
\item Shift the register one step to right with outputting the rightmost bit $a_0$;
\item Put $a_r=(\sigma \mod 2)$ into the leftmost of the register;
\item Put $\frac{\sigma-a_r}{2}$ into the memory;
\item Return to Step 2.
\end{enumerate}

{ The following result about  $2$-adic complexity of binary  sequences was firstly presented by Klapper et. al. \cite{Klapper:1997}.}
\begin{lemma}\label{le2} \cite{Klapper:1997}
\emph{(1)} Let $s$ be a periodic sequence generated by the
 FCSR with connection number $q$.
 Assume that $s=(s_0,s_1,
\cdots)$. Then, in $\mathbb{Q}_2$,  $\sum_{i=0}^{\infty}s_i2^i=\frac{p}{q}$, where $p$ is an integer
such that $-q\le p\le 0$.
Particularly, if $\gcd(p,q)=1$, then this FCSR is the shortest one which can produce $s$ and hence
$AC(s)=\lfloor\log (q+1)\rfloor$. \\
\emph{(2)} Conversely, let $s=(s_0, s_1,\cdots)$ be a binary periodic sequence.
If $\sum_{i=0}^{\infty}s_i2^i=\frac{p}{q}$ in $\mathbb{Q}_2$, then $s$
can be produced by the FCSR with connection number $q$.
\end{lemma}

Let $s$ be a periodic sequence and $\overline{s}$ its complementary sequence. It follows from
Lemma \ref{le2} and the fact $\sum_{i=0}^{\infty}s_i2^i+\sum_{i=0}^{\infty}\overline{s}_i2^i=\sum_{i=0}^{\infty}2^i=-1$ that
$s$ and $\overline{s}$ have the same $2$-adic complexity.  Hence when we refer to an ideal $2$-level autocorrelation sequence, we always assume that,  without loss of generality, its support set is an $(N, \frac{N+1}{2},\frac{N+1}{4})$ cyclic difference set.

\subsection{Linear feedback shift register}
An $r$-stage \emph{linear feedback shift register} (LFSR) over a finite field $\mathbb{F}_q$ is given in
Fig. \ref{f2}, where $q_i\ (1\le i\le r)\in\mathbb{F}_q$, $q_r\ne0$.
\begin{figure}[]
\begin{center}
\includegraphics[width=0.6\textwidth]{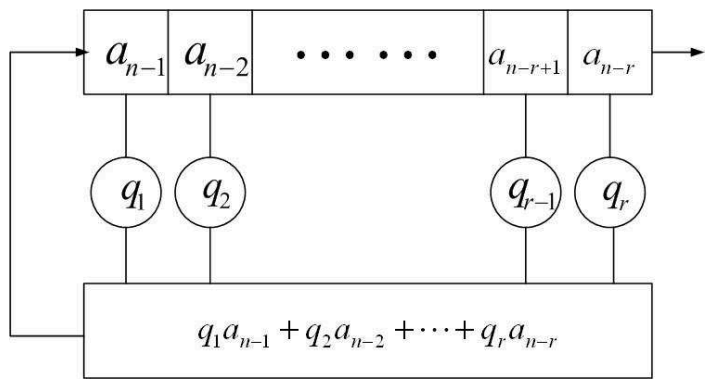}
\end{center}
\caption{Linear feedback shift register}\label{f2}
\end{figure}
We call $f(x)=\sum_{i=1}^{r}q_ix^i-1$   the \emph{connection polynomial} of this LFSR
and its  operation  is defined as follows:
\begin{enumerate}
\item Give an initial state $(a_{r-1}, a_{r-2},\cdots, a_0)$,
where $a_i\in\mathbb{F}_q$;
\item Compute a sum $\sigma=\sum_{i=0}^{r-1}q_ia_i$ over $\mathbb{F}_q$;
\item Shift the register one step to right with outputting the rightmost bit $a_0$;
\item Put $a_r=\sigma$ into the leftmost of the register;
\item Return to Step 2.
\end{enumerate}

The following is a well-known result on the linear complexity of periodic sequences.

\begin{lemma}\label{le7} \cite{Golomb:2005, Lichao:2010}
\emph{(1)} Let $s=(s_0,s_1,\cdots)$ be a periodic sequence generated by the LFSR with connection polynomial $f(x)$.
Then $\sum_{i=0}^{\infty}s_ix^i=\frac{g(x)}{f(x)}$.
Particularly, if $\gcd(g(x),f(x))=1$, then this LFSR is the shortest one which can produce $s$ and hence
$LC(s)=\deg(f(x))$. \\
\emph{(2)} Conversely, let $s=(s_0, s_1,\cdots)$ be a periodic sequence over $\mathbb{F}_q$.
If $\sum_{i=0}^{\infty}s_ix^i=\frac{g(x)}{f(x)}$, then $s$
can be produced by the LFSR with connection polynomial $f(x)$.
\end{lemma}

\subsection{Gauss sums}
Let $p$ be a prime and let $\psi$  be a multiplicative character of $\mathbb{F}_p$.
Define
$$G(\psi;\alpha)=\sum_{x\in\mathbb{F}_p^*}\psi(x)w_p^{\alpha x}$$
and
$$g(k;\alpha)=\sum_{x\in\mathbb{F}_p}w_p^{\alpha x^k},$$
where $k$ is an integer, $w_p=e^{\frac{2\pi i}{p}}$ is a $p$-th primitive unity of $\mathbb{C}$ and $\alpha\in\mathbb{F}_p$.
Both the above sums are called \emph{Gauss sums} and they are connected by the following results.
\begin{lemma}\cite{Berndt:1997}\label{le_gen1}
Let $\psi$ be a multiplicative character of $\mathbb{F}_p$ with order $k$. Then,
$$g(k;\alpha)=\sum_{j=1}^{k-1}G(\psi^j;\alpha)=\sum_{j=1}^{k-1}\psi^j(\alpha^{-1})G(\psi^j;1).$$
\end{lemma}

\begin{lemma}\cite{Berndt:1997}\label{le_gen2}
Assume that $p\equiv 1\mod 4$. One has \\
\emph{(1)} If $\psi$ is the quadratic character of $\mathbb{F}_p$, then $G(\psi;1)=g(2;1)=\sqrt{p}$;\\
\emph{(2)} If $\psi$ is a character of order $4$,
 then $$G(\psi;1)+G(\psi^3;1)=\pm\left\{2\left(\frac{2}{p}\right)(p+a\sqrt{p})\right\}^{1/2},$$
 where {$\left(\frac{2}{p}\right)\equiv 2^{\frac{p-1}{2}} \mod p$ is the Legendre symbol,} $a$ is an integer such that $a^2+b^2=p$, $a\equiv-\left(\frac{2}{p}\right)(\mod4)$;\\
 \emph{(3)} If $\psi$ is a nontrivial character, then $|G(\psi;1)|=\sqrt{p}$.
\end{lemma}

\section{ A new method of computing {the} $2$-adic complexity of binary sequences}

In this section, we will present a new method of computing {the} $2$-adic complexity of binary sequences.
The following is {a} key lemma of our method.

\begin{lemma}\label{le4}
Let $s=(s_0, s_1, \cdots, s_{N-1})$ be a binary  sequence with period $N$ and let $P_{ s }(x)=\sum_{i=0}^{N-1}s_ix^i\in \mathbb{Z}[x]$.
Let $A=(a_{i,j})_{N\times N}$
be the matrix defined by $a_{i,j}=s_{(i-j)\mod N}$,
and let us view $A$ as a matrix over $\mathbb{Q}$.
If $\det (A)\neq 0$, then there exist $u(x),\ v(x)\in\mathbb{Z}[x]$ such that

\begin{equation}\label{eq_ge}
  u(x)P_{ s }(x)+v(x)(1-x^N)=\det(A),
\end{equation}
where $\deg u\le N-1$, $\deg v\le N-2$.
\end{lemma}
\begin{proof}
It suffices to prove that
the following equation system has a  solution $(u_0,u_1,\cdots,u_{N-1},\\v_0,v_1,\cdots,v_{N-2})^{\text{T}}\in\mathbb{Z}^{2N-1}$, where $u_i$ and $v_i$ are the coefficients of $u(x)$ and $v(x)$ respectively.
\begin{equation}\label{eq3}
\left\{
\begin{aligned}
s_0u_0&+v_0&&=\det(A)\\
s_1u_0+s_0u_1&+v_1&&=0\\
&\vdots&&=\vdots\\
\sum_{i=0}^{N-2}s_{N-2-i}u_i&+v_{N-2}&&=0\\
\sum_{i=0}^{N-1}s_{N-1-i}u_i&&&=0\\
\sum_{i=1}^{N-1}s_{N-i}u_i&-v_0&&=0\\
\sum_{i=2}^{N-1}s_{N+1-i}u_i&-v_1&&=0\\
&\vdots&&=\vdots\\
s_{N-1}u_{N-1}&-v_{N-2}&&=0
\end{aligned}
\right.
.\end{equation}
The coefficient matrix $C$ of the above equation system is
\begin{displaymath}
C=\left(\begin{array}{cccccccc}
s_0 & 0 & \cdots & 0 & 1 & 0 & \cdots & 0 \\
s_1 & s_0 & \cdots & 0 & 0 & 1 & \cdots & 0 \\
\vdots & \vdots& \cdots & \vdots & \vdots & \vdots & \cdots & \vdots \\
s_{N-2} & s_{N-3} & \cdots & 0 & 0 & 0 & \cdots & 1 \\
s_{N-1} & s_{N-2} & \cdots & s_0 & 0 & 0 & \cdots & 0 \\
0 & s_{N-1} & \cdots & s_1 & -1& 0 & \cdots & 0 \\
\vdots & \vdots& \cdots & \vdots & \vdots & \vdots & \cdots & \vdots \\
0 & 0 & \cdots & s_{N-1} & 0& 0 & \cdots & -1
\end{array}\right).
\end{displaymath}
Adding the last $(N-1)$ rows of $C$ on the first $(N-1)$ rows, we get a new matrix
\begin{displaymath}
C'=\left(\begin{array}{cccccccc}
s_0 & s_{N-1} & \cdots & s_1 & 0 & 0 & \cdots & 0 \\
s_1 & s_0 & \cdots & s_2 & 0 & 0 & \cdots & 0 \\
\vdots & \vdots& \cdots & \vdots & \vdots & \vdots & \cdots & \vdots \\
s_{N-2} & s_{N-3} & \cdots & s_{N-1} & 0 & 0 & \cdots & 0 \\
s_{N-1} & s_{N-2} & \cdots & s_0 & 0 & 0 & \cdots & 0 \\
0 & s_{N-1} & \cdots & s_1 & -1& 0 & \cdots & 0 \\
\vdots & \vdots& \cdots & \vdots & \vdots & \vdots & \cdots & \vdots \\
0 & 0 & \cdots & s_{N-1} & 0& 0 & \cdots & -1
\end{array}\right).
\end{displaymath}
Then we have $\det(C)=\det(C')=\det(A)(-1)^{N-1}=\pm\det(A)\ne 0$.
Hence Equation (\ref{eq3}) has a unique solution $\alpha=(u_0,\cdots,u_{N-1}, v_0, \cdots, v_{N-2})^{\text{T}}=C^{-1}\beta$,
 where $\beta=(\det(A), 0,\cdots, 0)^{\text{T}}$. Noting that $C$ is a matrix over $\mathbb{Z}$ and $\det(C)= \pm\det(A)$,
 we have $\alpha=C^{-1}\beta\in \mathbb{Z}^{2N-1}$.
 We finish the proof.
$\hfill\blacksquare$\end{proof}

The following is our first main result on {the} 2-adic complexity of binary periodic sequence.

\begin{theorem}\label{th_gen}
Let the symbols be defined as in Lemma \ref{le4}. If $\gcd(1-2^N, \det(A))=1$,
then  $AC(s)=N$.
\end{theorem}
\begin{proof}
Since $\gcd(1-2^N, \det(A))=1$, we have $\det(A)\ne 0$.
According to Lemma \ref{le4}, there exist $u(x), v(x)\in\mathbb{Z}[x]$ such that
\begin{equation}
\label{eq1}
u(x)P_s(x)+v(x)(1-x^N)=\det(A).
\end{equation}
Substituting $x=2$ into the above equation and letting $M=P_s(2)$, we have
\begin{equation}
\label{eqM2N}
u(2)M+v(2)(1-2^N)=\det(A).
\end{equation}
Hence we have $\gcd(M, 1-2^N)=1$ since $\gcd(1-2^N, \det(A))=1$.
The result then follows from Lemma \ref{le2}.
$\hfill\blacksquare$\end{proof}

Before processing further discussions, we make two remarks on Theorem \ref{th_gen}.
 Firstly, let $d_1=\gcd(M, 1-2^N)$ and $ d_2=\gcd(1-2^N, \det(A))$.
 Then it follows from \eqref{eqM2N} that $d_2$ is divided by
$d_1$. Hence by Lemma \ref{le2}, the smallest connection number $q$ of $s$ is $\frac{1-2^N}{d_1}$,
which is lower bounded by $\frac{1-2^N}{d_2}$. Thus one can get a lower bound on $q$ and consequently a lower bound
 on the $2$-adic complexity of $s$ if $d_2=\gcd(1-2^N, \det(A))\neq 1$.
It is a more general result than Theorem \ref{th_gen}. However, for simplicity,
we would like to keep Theorem \ref{th_gen} as its present
form. Secondly, it is clear that Theorem \ref{th_gen} can be naturally generalized to 
$p$-ary sequences. However, we focus on binary sequences in the present paper.

Theorem \ref{th_gen} provides a new method to compute  {the} 2-adic complexity of binary sequences.
The key point of this method is to compute $\det (A)$ and then verify whether $\gcd(2^N-1, \det(A)) = 1$
, {where  $A$ is the circulant matrix constructed from
the sequence. According to linear algebra, $\det (A)$ can be computed as follows. }

\begin{lemma}\label{le_gen3}\cite{Davis:1994}
Let $s$ be a sequence with period $N$ and let $A=(a_{i,j})_{N\times N}$
be the matrix defined by $a_{i,j}=s_{(i-j)\mod N}$.
Then $\det(A)=\prod_{j=0}^{N-1} P_s(w_N^j)$, where $w_N=e^{\frac{2\pi i}{N}}$ is an $N$-th
primitive unity of $\mathbb{C}$.
\end{lemma}

It is clear that $P_s(1)=\sum_{i=0}^{N-1}s_i=|D_s|$. For $1\leq j\leq N-1$,
we have
$$P_s(w_N^j)=\sum_{i=0}^{N-1}s_i (w_N^j)^i=\sum_{i\in D_s} (w_N^j)^i.$$

Hence computing $P_s(w_N^j)$ is related to some exponential sums. If the corresponding  exponential sums can be
computed, then one can compute $\det(A)$ and check whether $\gcd(2^N-1, \det(A)) = 1$ holds.
This is the case  of Legendre Sequence, Ding-Helleseth-Lam Sequence and Ding-Helleseth-Martinsen Sequence,
as we will see in Section 4.2.
On the other hand, if the exponential sums can not be easily computed, we may use other methods to compute $\det(A)$.
This is the case of all the known
binary sequences with ideal 2-level autocorrelation, as we will see in Section 4.1.

\section{Determining 2-adic complexities of several binary sequences with Theorem \ref{th_gen} }

In this section, as applications of our new method, we will determine the 2-adic complexities
of many binary sequences. They are examples of the two cases discussed in the last section.

\subsection{All the known binary sequences with ideal $2$-level autocorrelation}

In this subsection, we will use
 Theorem \ref{th_gen}  to uniformly determine {the}
 2-adic complexities of all the known binary ideal $2$-level autocorrelation sequences.
Two lemmas will be needed. The first one is  a well-known result from linear
algebra.

\begin{lemma}\label{le8} \cite{Eves:1980}
Let $B=(b_{i,j})_{n\times n}$ be a matrix defined by
\begin{equation}
b_{i,j}=\left\{
\begin{aligned}
&x,\ \text{if } i=j;\\
&y,\ \text{if } i \ne j.
\end{aligned}
\right.\nonumber
\end{equation}
Then $\det(B)=(x+(n-1)y)(x-y)^{n-1}$.
\end{lemma}

\begin{lemma}\label{le5}
\emph{(1)} Let $p$ be an odd prime. If $q$ is a prime factor of $(2^p-1)$, then $q\ge (p+2)$.\\
\emph{(2)} Let $N=p(p+2)$, where $p$ and $p+2$ both are odd primes. If $q$ is a prime factor of $(2^N-1)$, then $q\ge (p+2)$.
\end{lemma}
\begin{proof}
We only give a proof for $(2)$. The proof for $(1)$ is similar and is left to the interested readers.
We regard $2$ as an element of $\mathbb{F}_q$, and denote by ord$(2)$
the order of $2$ in $\mathbb{F}_q^*$.

Since $2^{p(p+2)}\equiv1\mod q$, we have ord$(2)|p(p+2)$. Noting that {ord}$(2)\ne1$, therefore {ord}$(2)=p,\ p+2, \text{ or }N$.
Clearly, we also have {ord}$(2)|(q-1)$. Thus $q\geq  {p+2}$.
$\hfill\blacksquare$\end{proof}

Now we can introduce the  second main result.
\begin{theorem}\label{th_gen2}
Let $ s $ be any known ideal $2$-level autocorrelation sequence with period $N$. Then its $2$-adic complexity
is $N$.
\end{theorem}
\begin{proof}
By Theorem \ref{th_gen}, it suffices  to prove that $\gcd(1-2^N, \det(A))=1$, where  $A=(a_{i,j})_{N\times N}$
is the matrix defined by $a_{i,j}=s_{(i-j)\mod N}$.

Let $B=A^{\text T}A=(b_{i,j})_{N\times N}$, {where $A^{\text T}$ is the transpose of the matrix $A$.} Then $$b_{i,j}=\sum_{k=0}^{N-1}s_{k-i}s_{k-j}=\sum_{k=0}^{N-1}s_ks_{k+i-j}.$$
Thus $b_{i,j}=|D_s\cap (D_s+(i-j))|$.
Noting that $D_s$ is an $(N, \frac{N+1}{2}, \frac{N+1}{4})$ cyclic difference set,
we have
\begin{equation}
b_{i,j}=\left\{
\begin{aligned}
&\frac{N+1}{2},\ \text{if } i=j;\\
&\frac{N+1}{4},\ \text{if } i \ne j.
\end{aligned}
\right.\nonumber
\end{equation}
Hence, by Lemma \ref{le8} we have $\det(B)=(\frac{N+1}{2})^2(\frac{N+1}{4})^{N-1}$. Then $|\det(A)|=\sqrt{\det(B)}=\frac{N+1}{2}(\frac{N+1}{4})^{\frac{N-1}{2}}$.

According to Lemma \ref{le1}, there are only three cases for $N$.

If $N=2^n-1$, then $|\det(A)|=2^{n-1}2^{(n-{2})\frac{N-1}{2}}$. Since  $1-2^N$ is
odd, we have $\gcd(1-2^N, \det(A))=1$.

If $N=p$, then it follows from Lemma \ref{le5} that $\gcd(2^p-1, p+1)=1$. Hence
$\gcd(2^N-1, \det(A))=1$.

If $N=p(p+2)$, then $|\det(A)|=\frac{(p+1)^2}{2}(\frac{(p+1)^2}{4})^{\frac{N-1}{2}}.$
Similarly, it follows from Lemma \ref{le5} that $\gcd(p+1, 1-2^N)=1$ and $\gcd(1-2^N, \det(A))=1$.

We are done.
$\hfill\blacksquare$\end{proof}

{Theorem \ref{th_gen2} gives a uniform proof that all the known binary
sequences with ideal $2$-level autocorrelation have the maximum $2$-adic complexities.
To the authors' best knowledge,  the $2$-adic complexities of all these sequences except $m$-sequences are firstly determined.
Another  consequence of Theorem \ref{th_gen2} is that one can say more about the relation of linear complexity and $2$-adic complexity.
As we recalled, $m$-sequences are a class of sequences with minimum linear complexity and maximum $2$-adic complexity, while some
$l$-sequences are a class of sequences with minimum $2$-adic complexity and maximum linear complexity.
Now Legendre sequences, twin-prime sequences and Hall's sextic residue sequences are  examples of the sequences whose
 linear complexity and $2$-adic complexity both  attain the maximum.
}

\subsection{Legendre sequence and Ding-Helleseth-Lam sequence}

In this subsection, we will use Theorem \ref{th_gen} to determine 2-adic complexit{ies} of
Legendre sequence and Ding-Helleseth-Lam sequence. {According to Theorem \ref{th_gen} and the analysis followed,
we need to compute $P_s(w^j)$, which is related to some exponential sums.
For Legendre sequence, it is related to quadratic Gauss sum; while
 for Ding-Helleseth-Lam sequence, it is related to quartic Gauss sum. }

\begin{theorem}
Let $s$ be a Legendre sequence with period $p\equiv1\mod4$. Then $AC(s)=p$.
\end{theorem}
\begin{proof}
By Theorem \ref{th_gen}, it suffices  to prove that $\gcd(1-2^p, \det(A))=1$, where  $A=(a_{i,j})_{p\times p}$
is the matrix defined by $a_{i,j}=s_{(i-j)\mod p}$.

Let $w_p=e^{\frac{2\pi i}{p}}$, $B_0=\sum_{x\in D_0^{(2,p)}}w_p^x$ and $B_1=\sum_{x\in D_1^{(2,p)}}w_p^x$. According to the definition of Legendre sequence,
we have
\begin{equation}
P_s(w_p^j){=\sum_{i\in D_0^{(2,p)}}w_p^{ij}}
=\left\{\begin{array}{ll}
          \frac{p-1}{2}, & \text{ if } j=0; \\
          B_0, & \text{ if } j \in D_0^{(2,p)}; \\
          B_1, & \text{ if } j \in D_1^{(2,p)}.
        \end{array}
\right.
\nonumber\end{equation}

By Lemma \ref{le_gen2}, we have $1+2B_0=g(2;1)=\sqrt{p}$. Besides, one can easily deduce $1+B_0+B_1=0$.
Hence, $B_0=\frac{\sqrt{p}-1}{2}$ and $B_1=-\frac{\sqrt{p}+1}{2}$.
Thus, it follows from Lemma \ref{le_gen3} that
 \begin{equation}
 \begin{split}
 \det(A)&=\prod_{j=0}^{p-1}P_s(w_p^j)\\
 &=\frac{p-1}{2}\left(\frac{\sqrt{p}-1}{2}\right)^{\frac{p-1}{2}}\left(\frac{-\sqrt{p}-1}{2}\right)^{\frac{p-1}{2}}\\
 &=\frac{p-1}{2}\left(\frac{p-1}{4}\right)^{\frac{p-1}{2}}.
\end{split}\nonumber
\end{equation}
Similar argument as in Theorem \ref{th_gen2} show{s} that $\gcd(\det(A),2^p-1)=1$.
$\hfill\blacksquare$\end{proof}

Before introducing the result on the 2-adic complexity of Ding-Helleseth-Lam sequence,
we need a lemma.
\begin{lemma}\label{le_DHL}
Let $p\equiv1\mod4$ be a prime and $a$ be an odd integer such that $a^2+b^2=p$.
Then $\gcd(1\pm 2p+a^2p,2^p-1)=1$.
\end{lemma}
\begin{proof} We only prove that $\gcd(1+2p+a^2p,2^p-1)=1$ and the other case can be proved similarly.

Assume on the contrary that $\gcd(1+2p+a^2p,2^p-1)=d>1$. Let $r>0$ be an odd prime factor of $d$. Then one
can deduce that $r-1\equiv 0\mod p$ as in the proof of Lemma \ref{le5}. Thus $r=kp+1$, where
$k\geq 2$ is an even integer since both $r$ and $p$ are odd.
Let $1+2p+a^2p = ur$. Then $u$ is an even integer since $1+2p+a^2p$ is even. Clearly, $u\equiv 1\mod p$.
On the other hand, $u=\frac{1+2p+a^2p}{r}<\frac{1+2p+p^2}{r}<p+1$. Thus we get $u=1$
{which contradicts that $u$ is an even integer.}
Hence $\gcd(1+2p+a^2p,2^p-1)=1$.
$\hfill\blacksquare$\end{proof}

\begin{theorem}
Let $s$ be a Ding-Helleseth-Lam sequence with period $p\equiv1\mod4$.
{Then $AC(s)=p$.}
\end{theorem}
\begin{proof}
By Theorem \ref{th_gen}, it suffices  to prove that $\gcd(1-2^p, \det(A))=1$, where  $A=(a_{i,j})_{p\times p}$
is the matrix defined by $a_{i,j}=s_{(i-j)\mod p}$.

Let $\alpha$ be a primitive element of $\mathbb{F}_p$ and $w_p=e^{\frac{2\pi i}{p}}$.
Let $\lambda$ be a multiplicative character of $\mathbb{F}_p$ defined by $\lambda(\alpha)=i$.
Then the order of $\lambda$ is $4$.
{For $0\le i\le 3$, let $B_i=\sum_{x\in D_i^{(4,p)}}w_p^x$.} According to the definition of Ding-Helleseth-Lam sequence, we deduce
\begin{equation}\label{eq4}
P_s(w_p^j){=\sum_{i\in D_0^{(4,p)}\cup D_1^{(4,p)}}w_p^{ij}}
=\left\{\begin{array}{ll}
          \frac{p-1}{2}, & \text{ if } j=0; \\
          B_0+B_1, & \text{ if } j \in D_0^{(4,p)}; \\
          B_1+B_2, & \text{ if } j \in D_1^{(4,p)}; \\
          B_2+B_3, & \text{ if } j \in D_2^{(4,p)}; \\
          B_{3}+B_{0}, & \text{ if } j \in D_3^{(4,p)}.
        \end{array}
\right.
\end{equation}
Hence
\begin{equation}\label{DetLen}
\det(A) = \left(\frac{p-1}{2}\right)[(B_0+B_1)(B_1+B_2)(B_2+B_3)(B_3+B_0)]^{\frac{p-1}{4}}.
\end{equation}
{It follows from Lemma \ref{le5} that $\gcd(\frac{p-1}{2}, 1-2^p)=1$. }

By Lemma \ref{le_gen1}, we have
\begin{equation}
\left\{
\begin{aligned}
&1+4B_0=g(4;1)=G(\lambda;1)+G(\lambda^2;1)+G(\lambda^3;1);\\
&1+4B_1=g(4;\alpha)=\lambda(\alpha^{-1})G(\lambda;1)+\lambda^2(\alpha^{-1})G(\lambda^2;1)+\lambda^3(\alpha^{-1})G(\lambda^3;1);\\
&1+4B_2=g(4;\alpha^2)=\lambda(\alpha^{-2})G(\lambda;1)+\lambda^2(\alpha^{-2})G(\lambda^2;1)+\lambda^3(\alpha^{-2})G(\lambda^3;1);\\
&1+4B_3=g(4;\alpha^3)=\lambda(\alpha^{-3})G(\lambda;1)+\lambda^2(\alpha^{-3})G(\lambda^2;1)+\lambda^3(\alpha^{-3})G(\lambda^3;1).
\end{aligned}
\right.\nonumber
\end{equation}
One can easily verify that $G(\lambda^3;1)=\lambda(-1)\overline{G(\lambda;1)}$.
Noting that $\lambda(\alpha)=i$, the above equation can be reduced as
\begin{equation}\label{eq5}
\left\{
\begin{aligned}
&1+4B_0=g(4;1)=G(\lambda;1)+G(\lambda^2;1)+\lambda(-1)\overline{G(\lambda;1)};\\
&1+4B_1=g(4;\alpha)=-iG(\lambda;1)-G(\lambda^2;1)+i\lambda(-1)\overline{G(\lambda;1)};\\
&1+4B_2=g(4;\alpha^2)=-G(\lambda;1)+G(\lambda^2;1)-\lambda(-1)\overline{G(\lambda;1)};\\
&1+4B_3=g(4;\alpha^3)=iG(\lambda;1)-G(\lambda^2;1)-i\lambda(-1)\overline{G(\lambda;1)}.
\end{aligned}
\right.
\end{equation}

Let $R=\text{Re}(G(\lambda;1))$ and $I=\rm{Im}(G(\lambda;1))$.

If $p\equiv1\mod8$, then $\lambda(-1)={\lambda(\alpha^{\frac{p-1}{2}})=i^{\frac{p-1}{2}}}=1$. From Eq. (\ref{eq5}), {we get}
\begin{equation}
\left\{
\begin{aligned}
&2(B_0+B_1)=R+I-1;\\
&2(B_1+B_2)=I-R-1;\\
&2(B_2+B_3)=-R-I-1;\\
&2(B_3+B_0)=-I+R-1.
\end{aligned}
\right.
\end{equation}
Hence
\begin{equation}
\begin{split}
&16(B_0+B_1)(B_1+B_2)(B_2+B_3)(B_3+B_{0})\\
=&\left(1-(R+I)^2\right)\left(1-(R-I)^2\right)\\
=&1-2R^2-2I^2+(R^2-I^2)^2.
\end{split}
\end{equation}

It follows from Lemma \ref{le_gen2} that $R^2+I^2=p$ and $4R^2=2(p+a\sqrt{p})$. Hence we deduce $R^2=\frac{1}{2}(p+a\sqrt{p})$ and $I^2=\frac{1}{2}(p-a\sqrt{p})$.
Thus $$16(B_0+B_1)(B_1+B_2)(B_2+B_3)(B_3+B_0)=1-2p+a^2p.$$
{ It then follows from Eq. (\ref{DetLen}), $\gcd(\frac{p-1}{2}, 1-2^p)=1$ and Lemma \ref{le_DHL} that $\gcd(\det(A), 1-2^p)=1$. }

Similarly, if $p\equiv5\mod8$, then one can deduce
 $$16(B_0+B_1)(B_1+B_2)(B_2+B_3)(B_3+B_0)=1+2p+a^2p.$$
Hence we also have $\gcd(\det(A), 1-2^p)=1$.

{The proof is finished. }
$\hfill\blacksquare$\end{proof}

{
In this section, by using our new method, the 2-adic complexities of many
binary sequences with optimal autocorrelation are determined. We believe that it can be used to determine the 2-adic complexities of more binary
sequences. The reader is cordially invited to join this adventure.

On the other hand,  we must mention that this method has its own drawback.
It can not work for those binary sequences for which one has $\det(A)=0$,
 where $A$ is the circulant matrix defined by the sequence.
For example, let $s$ be a Ding-Helleseth-Martinsen sequence \cite{Cai:2009} with period $N=2q$, where
$q\equiv5\mod8$ is a prime. According to the definition of $s$, we have
$P_s(w_N^q)=P_s(-1)=0$, where $w_N=e^{\frac{2\pi i}{N}}$.
Then one can deduce  that $\det(A)=0$ from Lemma \ref{le_gen3}. Similarly, when $s$ is a Sidelnikov-Lempel-Cohn-Eastman sequence  \cite{Cai:2009} with period $N\equiv0\mod4$,
one can also prove that $\det(A)=0$.
Other methods may be needed to compute the $2$-adic complexities of these sequences.
}

\section{Observe binary sequences from different finite fields}
For a binary sequence $s$, since its elements consist of $0$ and $1$,
it can {also} be viewed as a sequence over { another
finite field. Let us denote by $LC_q(s)$ the linear complexity of
$s$ when we regard it as a sequence over finite field  $\mathbb{F}_q$.
Clearly, $LC_q(s)$  may be different when $q$ differs.
For example, let $s=11000,11000,\cdots$ be a binary sequence with period 5. Then one can verify that
$LC_2(s)=4$.
However, if we regard $s$ as a sequence over $\mathbb{F}_3$, then $LC_3(s)=5\neq 4$. }
It is natural to ask what is the relationship of the different linear complexity of the same binary sequence.
In this section, we will investigate this problem and will present some interesting results. To our knowledge, there are only
a few results about this problem; see \cite{Ding:2013}.

{Firstly, we have the following observation.
\begin{proposition}\label{Pr_qp}
Let $s$ be a binary periodic sequence  and $\mathbb{F}_q$ be a finite field with character $p$.
Then $LC_q(s)=LC_p(s)$.
\end{proposition}
\begin{proof}
Denote by $N$ the period of $s$. Then $\sum_{i=0}^{\infty}s_ix^i=\frac{P_s(x)}{1-x^N}$,
where $P_s(x)$ is the sequence polynomial of $s$. Since the greatest common divisor of
$P_s(x)$ and $1-x^N$ over $\mathbb{F}_q$ is equal to that of these two polynomials over $\mathbb{F}_p$,
the  result then follows  from Lemma \ref{le7}.$\hfill\blacksquare$
\end{proof}
 }

{Thanks to Proposition \ref{Pr_qp}, we will focus on the odd prime fields in the following.  }
Let $s$ be a binary sequence with period $N$. Now, view $P_s(x)$ and $1-x^N$ as polynomials in $\mathbb{Z}[x]$.
Let $g(x)=\gcd(P_s(x), 1-x^N)$ be a monic polynomial. Then $g(x)\in\mathbb{Z}[x]$.
Clearly, there exist polynomials $u(x),\ v(x)\in\mathbb{Z}[x]$ and a nonzero integer $a$ such that
 \begin{equation}\label{eq2}
u(x)P_s(x)+v(x)(1-x^N)=ag(x).
\end{equation}
{Note that $a\neq 1$ may hold since we are working not on the fields but on the rings. For example, let $P_s(x)=1+x$ and $N=5$
as in the before example. It is clear that $\gcd(P_s(x), 1-x^N)=1$ in $\mathbb{Z}[x]$.
Substituting  $P_s(x)=1+x$, $N=5$ and $g(x)=1$ into \eqref{eq2}, one gets
$u(x)(1+x)+v(x)(1-x^5)=a$, which will force $a$ to be even since both $1+x$ and $1-x^5$ are even
 if $x$ is an odd integer. Hence $a\neq 1$.
}

\begin{theorem}\label{th2}
{Let  $p$ be a prime.} The other notations are the same as defined
 in the above paragraph. {Then $LC_p(s)\le N- \deg g(x)$.}
  If $p\not|a$, then the equality holds.
\end{theorem}
\begin{proof}
{Let us} view $P_s(x)$ and $1-x^N$ as polynomials in $\mathbb{F}_p[x]$ { and denote by } $d(x)$  their greatest common
divisor in $\mathbb{F}_p[x]$.
If we also view $g(x)$ as a polynomial over $\mathbb{F}_p$, then $g(x)|d(x)$. Hence, by Lemma \ref{le7}, we deduce that
$s$ can be generated by the LFSR over $\mathbb{F}_p$ with connection polynomial $(1-x^N)/g(x)$. Therefore,
$LC_p(s)\le N-\deg g(x)$.

If $p\not|a$, then $a\ne0$ in $\mathbb{F}_p$. By Equation (\ref{eq2}), we have $d(x)|g(x)$, which means $d(x)=g(x)$. Thus, it follows from Lemma \ref{le7} that $LC_p(s)=N-\deg g(x)$.
$\hfill\blacksquare$\end{proof}

We should remind  the reader that the inequality in the above theorem holds sometimes. For example,
let $s=(11010)$ be a sequence of period $5$. Then we have  $g(x)=\gcd(P_s(x), 1-x^5)=1$ in $\mathbb{Z}[x]$
while $d(x)=\gcd(P_s(x),1-x^5)=x-1$ in $\mathbb{F}_3[x]$. Hence $LC_3(s)=4<5=N-\deg g(x)$.


\begin{corollary}
Let $s$ be a binary ideal $2$-level autocorrelation sequence with period $N$, and let {
$p$ be an odd prime.

(1) If $|D_s|=\frac{N+1}{2}$ and $p\not|(N+1)$, then $LC_p(s)=N$;

(2) If $|D_s|=\frac{N-1}{2}$, $p\not|(N+1)$ and $p|(N-1)$, then $LC_p(s)=N-1$;

(3) If $|D_s|=\frac{N-1}{2}$ and $p\not|(N^2-1)$, then $LC_p(s)=N$.

%

}
\end{corollary}

\begin{proof}

(1) Assume that $|D_s|=\frac{N+1}{2}$. Then $D_s$ is an $(N, \frac{N+1}{2}, \frac{N+1}{4})$ cyclic difference set.
In the proof of
Theorem \ref{th_gen2}, it is proved  that $\det(A)=\pm (\frac{N+1}{2})(\frac{N+1}{4})^{(\frac{N-1}{2})}$,
where $A=(a_{ij})=(s_{(i-j)\mod N})$ be the matrix defined by $s$.
Comparing  \eqref{eq_ge} and \eqref{eq2}, one has $g(x)=1$ and $a=\det(A)$. Hence $p\not|a$ by the assumption $p\not|(N+1)$.
The result then follows from Theorem \ref{th2}.

(2) Assume that $|D_s|=\frac{N-1}{2}$.
Then $D_{\overline{s}}$ is an $(N, \frac{N+1}{2}, \frac{N+1}{4})$ cyclic difference set.
According to the result of the first part, we have $\gcd(P_{\overline{s}}(x),1-x^N)=1$. Noting that $P_s(x)+P_{\overline{s}}=(1-x^N)/(1-x)$,
one can deduce $\gcd(P_s(x),(1-x^N)/(1-x))=1$. Because $p|(N-1)$, we have $P_s(1)=0$ which means $(1-x)|P_s(x)$.
Therefore $\gcd(P_s(x)/(1-x), (1-x^N)/(1-x))=1$. The result then follows from Lemma \ref{le7}.

(3) One can deduce the result similarly as the second part.
$\hfill\blacksquare$\end{proof}

Since $N$ is finite, the number of primes dividing  $N^2-1$ is finite.
Hence except finitely many cases,
 the linear complexity of  a  binary ideal $2$-level autocorrelation sequence
  regarded as a sequence over another prime finite field attains the maximum.

The following interesting result follows immediately from the above corollary.
\begin{theorem}
Let $s$ be a binary ideal $2$-level autocorrelation
sequence with period $N=2^n-1$. Let $\mathbb{F}_q$ be a finite field with an odd character $p$.

(1) If $|D_s|=\frac{N+1}{2}$, then $LC_q(s)=N$;

(2) If $|D_s|=\frac{N-1}{2}$  and  $p|(N-1)$, then $LC_q(s)=N-1$;

(3) If $|D_s|=\frac{N-1}{2}$ and $p\not|(N-1)$, then $LC_q(s)=N$.

\end{theorem}

\section{Conclusion}
{To summarize, the contributions of this paper are threefold. Firstly,
a new method is presented to compute the $2$-adic complexity of binary sequences.
Secondly, all the known binary sequences with ideal $2$-level autocorrelation
are uniformly proved to have the maximum $2$-adic complexities, i.e. their $2$-adic complexities equal their periods.
As far as the authors known,  the $2$-adic complexities of all these sequences except $m$-sequences are not known before this paper.
We also  investigated the $2$-adic complexities of two classes of optimal autocorrelation sequences with period $N\equiv1\mod4$.
Thirdly, the new method is used to study the linear complexity
of binary sequences taken as sequences over other finite fields.
 An interesting finding is  that, except finitely many cases,
 the linear complexity of a  binary ideal $2$-level autocorrelation sequence
  regarded as a sequence over another prime finite field attains the maximum.


 }

\section*{Acknowledgments}
The work was supported by the Natural Science Foundation of China
(No. 61272484) and Basic Research Fund of National University of Defense Technology (No. CJ 13-02-01).

\bibliographystyle{model1a-num-names}
\bibliography{<your-bib-database>}

\end{CJK*}
\end{document}